\documentclass[letterpaper,twocolumn]{article}
\usepackage[letterpaper, margin=0.8in, bottom=1in]{geometry}

\title{Ariadne and Theseus: Exploration and Rendezvous \\with Two Mobile Agents in an Unknown Graph}

\author{Romain Cosson \\ \texttt{romain.cosson@inria.fr}}
\date{}

\usepackage[english]{babel}
\usepackage{algorithm}
\usepackage{algpseudocode}
\usepackage{amsmath}
\usepackage{cases}
\usepackage{graphicx}
\usepackage[dvipsnames]{xcolor}   
\usepackage[utf8]{inputenc} 
\usepackage{url}            
\usepackage{booktabs}       
\usepackage{amsfonts}       
\usepackage{nicefrac}       
\usepackage{microtype}      
\usepackage{tikz}
\usetikzlibrary{arrows,topaths}
\usepackage{mathtools}
\usepackage{amssymb}
\usepackage{amsthm}
\usepackage{amsmath,amssymb}
\usepackage{comment}        
\usepackage{enumitem}  
\usepackage{cite} 
\usepackage{parskip}

\usepackage{natbib}
\usepackage[colorlinks=true, allcolors=blue]{hyperref}

\newtheorem{theorem}{Theorem}[section]
\newtheorem{proposition}[theorem]{Proposition}

\newtheorem{remark}[theorem]{Remark}

\newtheorem{claim}[theorem]{Claim}

\newtheorem*{definition*}{Definition}

\newtheorem*{theorem*}{Theorem}

\newtheorem*{lemma*}{Lemma}

\newcommand{\Ocal}{\mathcal{O}}

\newcommand{\Nbb}{\mathbb{N}}

\newcommand{\ceil}[1]{{\lceil #1 \rceil}}

\begin{document}
\maketitle

\begin{abstract}
We investigate two fundamental problems in mobile computing: exploration and rendezvous, with two distinct mobile agents in an unknown graph. 
The agents may communicate by reading and writing information on whiteboards that are located at all nodes. 
They both move along one adjacent edge at every time-step. 
In the exploration problem, the agents start from the same arbitrary node and must traverse all the edges. 
We present an algorithm achieving collective exploration in $m$ time-steps, where $m$ is the number of edges of the graph. 
This improves over the guarantee of depth-first search, which requires $2m$ time-steps. 
In the rendezvous problem, the agents start from different nodes of the graph and must meet as fast as possible. 
We present an algorithm guaranteeing rendezvous in at most $\frac{3}{2}m$ time-steps. 
This improves over the so-called `wait for Mommy' algorithm which is based on depth-first search and which also requires $2m$ time-steps. 
Importantly, all our guarantees are derived from a more general asynchronous setting in which the speeds of the agents are controlled by an adversary at all times. 
Our guarantees generalize to weighted graphs, when replacing the number of edges $m$ with the sum of all edge lengths. 
We show that our guarantees are met with matching lower-bounds in the asynchronous setting.
\end{abstract}


\section{Introduction}
In 1952, Claude Shannon presented an electromechanical mouse capable of finding the exit of a maze embedded in a $5\times 5$ grid. The device was baptised `Theseus' in reference to the mythical hero who must escape an inextricable labyrinth after having killed the ferocious Minotaur. Shannon's mouse is arguably the first example of an autonomous mobile device
\cite{klein2018} and it inspired a number of micro-mouse competitions globally.  

The algorithm used by Shannon's mouse is known today as `depth-first search' (DFS). Its analysis dates back to the 19th-century, making it one of the few algorithms preceding the era of computers \cite{lucas1883recreations}. Depth-first search generalizes the ancient maze-solving heuristic `right-hand-on-the-wall' which can be used in the absence of cycles, i.e. for trees. Given the ability to mark edges (e.g. with a chalk), an agent using depth-first search is guaranteed to traverse each edge once in both directions and then return to the origin. 
In modern terms, we say it achieves graph exploration in $2m$ moves, where $m$ is the number of edges of the graph. The algorithm is optimal in the sense of competitive analysis \cite{miyazaki2009online}.

\subparagraph{Main results.} In the myth, Theseus can count on the help of the ingenious Ariadne. In this paper, we start by studying the question of whether two agents can solve a maze faster than a single agent. 
We answer by the affirmative using the formalism of collective exploration introduced by \cite{fraigniaud2006collective}.
Specifically, our main contribution to this problem is a collective graph exploration algorithm for two agents which requires exactly $m$ time-steps to explore any graph with $m$ edges. It is the first method to provide a quantitative improvement over the guarantee provided by a single depth-first search, answering a question of \cite{brass2014improved}. The algorithm is presented in Section \ref{sec: 2agents}. 

We then consider the problem of `rendezvous' in which the two agents start from different nodes and must meet somewhere in the graph. 
While the problem has attracted a rich body of literature (see the survey of \cite{pelc2019deterministic}) the setting where the graph is unknown and the agents are distinguishable (i.e. they may use different algorithms) has surprisingly received little attention. The best method at hand for this problem is the simple `Wait for Mommy' algorithm in which one agent stays immobile while the other performs a depth-first search, achieving rendezvous in at most $2m$ time-steps. 
Our contribution to this problem is an algorithm achieving rendezvous of two mobile agents in only $\ceil{\frac{3}{2}m}$ time-steps. The algorithm is presented in Section \ref{sec:rdv}. 

We formalize in Section \ref{sec:setting} an asynchronous navigation model in which an adversary chooses at each round the agent which may perform a move. Our guarantees for the synchronous setting are directly derived from more general results which hold for this asynchronous setting. We also introduce the practical formalism of `navigation tables' which could be applicable to other aspects of mobile computing. 

We extend in Section \ref{sec:discussion} the generality of our model of asynchrony to the case where the agents move continuously. We also show that all our guarantees remain valid for weighted graphs, if the number of edges $m$ is replaced by the sum of all edge lengths $L$. Finally, we provide lower-bounds (on the cumulative moves of the agents) of $2L$ for collective exploration and $3L$ for rendezvous, matching the guarantees of our algorithms in this setting. 

\subsection{Related works}
The model of mobile computing considered in this paper corresponds exactly to the classical description of depth-first search \cite{lucas1883recreations}. The graph is unlabelled, the agents have severely limited memory (enough to recall the edge with which they enter a node) and can mark the endpoints of edges (thereafter called \textit{ports} or \textit{passages}) with a constant number of symbols or \textit{makers}. 

\subparagraph{History of Depth-First Search.} 
The first formal presentation of depth-first search for graphs is designed by Trémaux in response to an open problem asking, in eloquent terms, whether there exists a deterministic algorithm for solving mazes: 
\begin{quote}
\textit{Reader, imagine yourself lost in the crossroads of a labyrinth, in the galleries of a mine, in the quarries of the catacombs, under the shady alleys of a forest. You don't have the thread of Ariadne in your hand, and you are in the same situation as Little Tom Thumb after the birds have eaten the breadcrumbs. What can you do to find your way back to the entrance of the labyrinth?
}
\end{quote}
\hfill Edouard Lucas, \textit{The Game of Labyrinths}, \cite{lucas1883recreations}.
\medskip

Trémaux's Depth-First Search can synthetically be described as follows (see Algorithm \ref{alg:tremaux} for details). The searcher (say, Ariadne) always attempts to go through a passage that she has not traversed, but prefers to backtrack rather than to cross her own path. When she is at a node of which all passages have been explored, she exits this node using the passage by which she first discovered it. Shortly after Trémaux, \cite{tarry1895probleme} observed that a slightly different algorithm entails the same guarantees. The specificities of  this variant are discussed in \cite{even2011graph}.

Later, depth-first search received a lot of attention outside mobile computing, due to the growing body of research on other aspects of theoretical computer science (see e.g. \cite{golomb1965backtrack}). The algorithm is usually implemented by stacking all the neighbours of the last discovered node, and applying the method recursively until all nodes have been discovered (and the stack is empty). This implementation leads to a `depth-first search ordering' which corresponds to the order in which the nodes of a graph would be discovered by Trémaux's algorithm ; however it does not illustrate well the application of the method to mazes, since it implicitly assumes that the searcher can `jump' between previously discovered nodes, just like a computer can jump between different addresses in memory in constant time (RAM model). 
Depth-first search orderings are key ingredients of so-called `linear algorithms' for graphs, the study of which was initiated by \cite{tarjan1972depth} to compute the strongly connected components of a directed graph and the biconnected components of an undirected graph. Surprisingly, some fundamental questions about depth-first search orderings remain open today \cite{aggarwal1987random}.

\subparagraph{Single mobile agent in an unknown undirected graph.} We now recall more recent results on the topic of graph exploration with a single agent, and refer to the survey of \cite{das2019graph} for more details. Specifically, we observe that the literature usually varies in the three following modelling choices. 

\textit{Storage.} The terms refers to maximum amount of information that can be stored at any node of the graph (e.g. on a whiteboard) and is quantified in bits. The depth-first search algorithm of \cite{lucas1883recreations} requires $\Ocal(\Delta)$ bits of storage, where $\Delta$ is the maximum degree of the graph. This is because the searcher must store on each node the list of adjacent port numbers that have been explored. Universal exploration sequences, introduced by \cite{reingold2008undirected}, allow to explore a graph with no storage (and $\Ocal(\log(n))$ memory) but only come with a poly($n$) runtime, where $n$ is the number of nodes in the graph. 

\textit{Memory.} The term refers to the maximum on-board memory that the agent may use during navigation. It is also quantified in bits. In a labeled graph (i.e. nodes have a unique identifier\footnote{there is no deterministic exploration algorithm using no memory and no storage for unlabelled graphs.}), $\Ocal(m)$ bits of memory are sufficient to encode the graph, and thus to implement DFS without storage. In the storage-based presentation of depth-first search, it is implicitly assumed that the agent uses $\Ocal(\log \Delta)$ memory because it recalls the edge by which it enters a given node. 
The `rotor router' algorithm allows to explore undirected graphs with zero memory (and $\Ocal(\log \Delta)$ storage) but it requires $\Ocal(Dm)$ steps to go through all edges, where $D$ is the diameter of the graph (i.e. the maximum distance between any two nodes) \cite{yanovski2003distributed,menc2017time}.

\textit{Feedback.} The term refers to the amount of information revealed to the agent when it attains a new node. In the original presentation of depth-first search and in the setting of this paper, an agent gets to see only the port numbers of the edges adjacent to its position (i.e. the agent is short-sighted). A line of work initiated by \cite{kalyanasundaram1994constructing} studies the case where the entire neighbourhood of a node is revealed to the agent when that node is visited. 
This problem is still actively studied, see e.g. \cite{megow2012online,birx2021improved,baligacs2023exploration,akker2024multi}. Another line of work initiated by \cite{papadimitriou1991shortest,fiat1998competitive} considers the case where the set of nodes is partitioned in \textit{layers} and the agent gets to see all nodes and edges adjacent to a layer when it attains some node of that layer. This problem is called `layered graph traversal' and is also subject to recent research \cite{bubeck2022shortest,bubeck2023randomized}.  

In this paper, we shall focus on algorithms which rely on the same (natural) assumptions as the original setting of depth-first search \cite{lucas1883recreations}. Therefore, the agents (Ariadne and Theseus) use $\Ocal(\Delta)$ storage per node, they have $\Ocal(\log \Delta)$ memory (just enough to remember the edge from which they arrive), and no additional feedback (short-sighted agent). The graph is undirected and unlabelled.



\subparagraph{Collective exploration.} Collective exploration studies algorithms for exploring unknown environments with multiple mobile agents. The setting was introduced by \cite{fraigniaud2006collective} for the special case of trees and easily adapts to general graphs \cite{brass2011multirobot}.

Consider an unknown graph $G=(V,E)$.  A team of $k\in \Nbb$ agents initially located at a same node is tasked to traverse all edges of the graph. The agents move synchronously along one adjacent edge at each round. An edge is revealed only when one agent becomes adjacent to that edge.  For an exploration algorithm \texttt{ALG}, we denote by $\texttt{ALG}(G,k)$ the number of rounds the team takes to traverse all edges of $G$. The literature mainly focuses on the analysis of the competitive ratio of \texttt{ALG} defined by,
\begin{equation*}
    \texttt{competitive-ratio}(\texttt{ALG},k) = \max_{G} \frac{\texttt{ALG}(G,k)}{\texttt{OPT}(G,k)},
\end{equation*}
where $\texttt{OPT}(G,k)$ denotes the minimum amount of rounds required by the team to go through all edges of $G$ if the graph was known in advance.

The analysis of the competitive ratio has been quite fruitful for the special case of trees (i.e. when $G$ is assumed to have no cycle). To date, the best competitive ratio is in $\Ocal(k/\log(k))$ if the agents use distributed communication (i.e. they communicate only through storage) \cite{fraigniaud2006collective} and is in $\Ocal(\sqrt{k})$ if the agents are allowed complete communication (i.e. they are controlled by one central algorithm) \cite{cosson2024collective}. Several other collective exploration algorithms have been proposed for trees, e.g. \cite{brass2011multirobot,ortolf2014recursive,cosson2024breaking}. 

On the contrary, very little is known about collective exploration of general graphs. A guarantee of $2m/k + 2n(\ln k +1)$ was proposed by \cite{brass2014improved}, and the case where the number of robots is very large $k\geq Dn^{1+\epsilon}$ for some $\epsilon >0$, was studied by \cite{dereniowski2015fast}. To date, the best competitive ratio remains in $2k$ and is attained by the trivial algorithm that uses a single agent to perform depth-first search and keeps the $k-1$ remaining agents idle at the origin. It was highlighted as an open question by \cite{brass2014improved} whether this guarantee could be improved. This paper answers affirmatively by showing that two agents may go trough all edges of a graph in $m$ rounds, thus improving the competitive ratio from $2k$ to $k$, for any $k\geq 2$ (see discussion in Section~\ref{sec:synchronous}). We note that some special cases of collective graph exploration have also been considered in the literature, such as grid graphs with rectangular obstacles \cite{ortolf2012online, cosson2023efficient} and cycles \cite{higashikawa2014online}.

Finally, we observe that collective exploration can be generalized to the asynchronous setting (i.e. when agents have adversarial speeds) and to weighted graphs (i.e. where the weights represent a cost associated to traversing an edge). Our algorithm and its guarantee adapt to these generalizations. 

\subparagraph{Rendezvous of two mobile agents.} We provide a short overview of the rendezvous problem, and refer to the survey of \cite{pelc2019deterministic} for more details. We note that most of the effort on the problem has been on the question of feasibility of rendezvous, rather than on the runtime (see e.g. \cite{guilbault2011asynchronous}). 
The reason why rendezvous might be infeasible even when the graph is known by the agents, is that it is generally assumed that the agents are indistinguishable (except perhaps by some personal label or by their initial position in the graph) and that they must use the same deterministic algorithm. 
The most obvious example of an infeasible rendezvous is that of two identical agents in the ring \cite{kranakis2003mobile} for which there is no way to break the symmetry. The problem of rendezvous naturally generalizes to asynchronous models of mobile computing \cite{de2006asynchronous,czyzowicz2012meet,dieudonne2013meet}.
In contrast to most previous works, we consider the situation where the agents have distinct algorithms and can communicate by reading and writing on the whiteboards at all nodes (i.e. with storage). This setting is usually treated in the rendezvous literature by the aforementioned `Wait for Mommy' algorithm \cite{pelc2019deterministic}, which requires $2m$ synchronous time steps. Our rendezvous algorithm improves this quantity to $\frac{3}{2}m$. We note that replacing depth-first search by the algorithm of \cite{panaite1999exploring} in the `Wait for Mommy' algorithm leads to a rendezvous in $m+3n$, which is better for graphs with a super-linear number of edges, but this alternative does not provide satisfactory guarantees for weighted graphs and does not deal with asynchrony.

\section{Problem setting and definitions}\label{sec:setting}
\subsection{Navigation model}\label{sec:navigation-model}
The algorithms presented in this paper rely on the assumption that it is possible to read and write information on whiteboards located at all nodes. It will be more convenient to assume that there are multiple (smaller) whiteboards at each node, one for each \textit{passage} adjacent to the given node. A passage (more commonly referred to as a \textit{port}) is defined as the intersection of an edge with one of its endpoints. The denomination is borrowed from \cite{even2011graph}. The agent can choose between a finite number of \textit{markers} to leave information on a passage. Obviously, an exploration algorithm requiring $\Ocal(1)$ storage on each passage entails an associated exploration algorithm requiring $\Ocal(\Delta)$ storage on each node, where $\Delta$ is the maximum degree of the graph, as announced in the introduction. 

\subparagraph{Move of an agent.} The \textit{move} of an agent can be decomposed in the following steps: 
\begin{enumerate}[label=S\arabic*]
    \item The agent reads at its location $u$, it decides whether exploration continues, and if so it chooses a port (or passage) at $u$ denoted $p_u$; \label{step1}
    \item The agent may change the marker of $p_u$;  \label{step2}
    \item The agent uses $p_u$ to traverse the chosen edge;  \label{step3}
    \item The agent arrives at $v$ through port $p_v$, it reads the markers adjacent to $v$;  \label{step4}
    \item The agent possibly changes the marker of $p_v$.  \label{step5}
\end{enumerate}
Note that this decomposition is directly inspired from the simple description of Trémaux's algorithm found in \cite{even2011graph}. In particular, it would not be possible to implement depth-first search, or any linear exploration algorithm, without steps \eqref{step4} and \eqref{step5} (see. \cite{menc2017time}). 

\subparagraph{Synchronous and asynchronous models.} We study two models of mobile computing, the \textit{synchronous model} and the \textit{asynchronous model}. 

\textit{Synchronous model.} In the synchronous model, both agents have one move at each time-step. In particular, they perform steps \eqref{step1} to \eqref{step5} simultaneously. If both agents start the round on the same node, we further assume that they can entirely communicate and coordinate during \eqref{step1} to avoid conflicting choices.

\textit{Asynchronous model.} In the asynchronous model, we assume that an adversary decides at each round which of the two agents will have a move. That agent then performs all steps \eqref{step1} to \eqref{step5} without interruption. The generality of this model of asynchrony is explained in Section \ref{sec:discussion}. In particular, this model will encapsulate the situation where the adversary decides at all times the (continuous) speeds of the agents.

\subsection{Trémaux's algorithm}
We now turn to the formal description of Trémaux's algorithm (Algorithm~\ref{alg:tremaux}) in the model of mobile computing defined above. The agent is able to mark passages using \textit{markers} E (Explored), F (First entry), and B (Backtrack) in place of the default marker $\emptyset$ (unmarked). For the sake of self-containedness, we provide a brief analysis of Trémaux's algorithm and refer to \cite{even2011graph} for more details.  

In pseudo-code and proofs, we shall use the notation $u \xrightarrow{e} v$ to denote the passage (i.e. port) of edge $e=(u,v)\in E$ at node $u\in V$. The value of $v$ becomes known to an agent at $u$ only if it chooses to move along this passage.  

\begin{algorithm}[ht]
\caption{\texttt{DFS}, Trémaux's algorithm}\label{alg:tremaux}
\begin{algorithmic}[1]
\If{there is a passage $u \xrightarrow{e} v$ marked B}
\State Mark $u \xrightarrow{e} v$ by E and traverse $e$
\ElsIf{there is an unmarked passage $u \xrightarrow{e} v$}
\State Mark $u \xrightarrow{e} v$ by E and traverse $e$
\If{$v$ was previously discovered}
\State Mark $v \xrightarrow{e} u$ by B
\Else
\State Mark $v \xrightarrow{e} u$ by F
\EndIf
\ElsIf{there is a passage $u \xrightarrow{e} v$ marked F}
\State Traverse $e$
\Else
\State Declare STOP
\EndIf
\end{algorithmic}
\end{algorithm}

\begin{proposition}[\cite{lucas1883recreations}]\label{prop:tremeaux} For any graph $G$ with $m$ edges, Trémaux's algorithm (Algorithm~\ref{alg:tremaux}) has the agent traverse all edges once in each direction and stop at the origin. 
\end{proposition}
\begin{proof}
The proof of the result relies on Claims \ref{claim:t1} and~\ref{claim:t2}, which are shown below.

\begin{claim}\label{claim:t1}
When Algorithm~\ref{alg:tremaux} terminates, the agent is at the origin and explored all edges.
\end{claim}
\begin{proof}
    The algorithm terminates only when the agent is adjacent to passages marked E. All discovered nodes other than the origin have a passage marked F. Thus the algorithm terminates when the agent is at the origin. 
    Now, assume by contradiction that some edge was never traversed, and consider without loss of generality one such edge that is adjacent to a node that has been discovered. 
    Note that the corresponding node cannot be the origin. Consider the last move when the agent left that node. That move must have been through an edge $e$ with endpoints marked by E and F. Iterating this reasoning, there is a sequence of directed edges with markers E$ \rightarrow $F leading from this node to the origin, where the agent is presently (the last edge of the sequence is possibly marked E$ \rightarrow $B). This is a contradiction because the origin is only adjacent to markers E.
\end{proof}
\begin{claim}\label{claim:t2}
No edge can be traversed twice in the same direction.
\end{claim}
\begin{proof}
Assume by contradiction that some edge $e$ is traversed twice in the same direction, and consider the first occurrence of such event on passage $u \xrightarrow{e} v$. From the description of the algorithm, it is clear that this passage was marked F. Now, observe that the edges with one endpoint marked F form a directed sub-tree of the explored graph rooted at the origin (the tail of the arrow is marked E and the head is marked F). If some edge of this sub-tree is traversed twice in the reverse direction, it must be that this same edge was traversed twice in the other direction, thereby contradicting the fact that we chose the first occurrence of the event of an edge being traversed twice in the same direction.  
\end{proof}
The claims above finish the proof of Proposition \ref{prop:tremeaux}.
\end{proof}
\subsection{Navigation table}\label{sec:nav-table}
We now provide an equivalent and more compact way to describe a navigation algorithm -- such as Trémaux's algorithm -- in a table with four columns, (see e.g. Table~\ref{table:Trémaux}). This format will be useful to give a clear and compact description of our algorithms. The lines are ordered by priority and the agent chooses the first line that fits its current situation.
The first column ($p_u$) denotes the markers available at $u$, the initial position of the agent, and defines the first step \eqref{step1} of a move. 
For instance, we can read from the first column of Table~\ref{table:Trémaux} that Trémaux's algorithm always prefers an adjacent port marked $\emptyset$ to one marked F. 
We can also infer that if the agent is not adjacent to a marker in \{B, $\emptyset$, F\}, the algorithm stops. The second column $p_u'$ tells how the marker should be changed in step \eqref{step2}. 
The third column $p_v$ indicates what the robot reads at $v$, the new position of the agent \eqref{step4}. The forth column $p_v'$ prescribes how the marker of edge $(u,v)$ at $v$ should change \eqref{step5}. The symbol dash `-' signals that all markers are accepted, or unchanged, in the corresponding step.

\begin{table}[h!]
\centering
\begin{tabular}{||c c c c||} 
 \hline
    $p_u$ & $p_u'$ & $p_v$ & $p_v'$ \\ [0.5ex] 
 \hline\hline
 B & E & - & - \\ 
 $\emptyset$ & E & $v$ discovered & B \\
 $\emptyset$ & E & $v$ undiscovered & F \\
 F & - & - & -\\
 \hline
\end{tabular}
\caption{Trémaux's navigation Table.}
\label{table:Trémaux}
\end{table}

\begin{figure}
    \centering
    
\begin{tikzpicture}[
    thick,
    main node/.style={circle, fill=blue!20, draw, font=\sffamily\large\bfseries}
  ]

  \node[main node] (A) at (0,0) {A};
  \node[main node] (B) at (2,0) {B};
  \node[main node] (C) at (1,-2) {C};
  \node[main node] (D) at (4,0) {D};
  \node[main node] (E) at (3,-2) {E};

  \draw (A) -- (B) -- (D);
  \draw (A) -- (C) -- (B);
  \draw (C) -- (E);

  \draw[blue, dashed, ->] (A) edge[bend right=10] node[midway, below, font=\small] {1} (B);
  \draw[blue, dashed, ->] (B) edge[bend right=10] node[midway, below,font=\small] {2} (D);
  \draw[blue, dashed, ->] (D) edge[bend right=10] node[midway, above,font=\small] {3} (B);
  \draw[blue, dashed, ->] (B) edge[bend right=10] node[midway, left,font=\small] {4} (C);
  \draw[blue, dashed, ->] (C) edge[bend right=10] node[midway, right,font=\small] {5} (A);
  \draw[blue, dashed, ->] (A) edge[bend right=10] node[midway, left,font=\small] {6} (C);
  \draw[blue, dashed, ->] (C) edge[bend right=10] node[midway, below,font=\small] {7} (E);
  \draw[blue, dashed, ->] (E) edge[bend right=10] node[midway, above,font=\small] {8} (C);
  \draw[blue,dashed, ->] (C) edge[bend right=10] node[midway, right,font=\small] {9} (B);
  \draw[blue, dashed, ->] (B) edge[bend right=10] node[midway, above,font=\small] {10} (A);
\end{tikzpicture}

\caption{Example of execution of Trémaux's algorithm, starting from node A, for a graph with $m=5$ edges.}
    \label{fig:enter-label}
\end{figure}
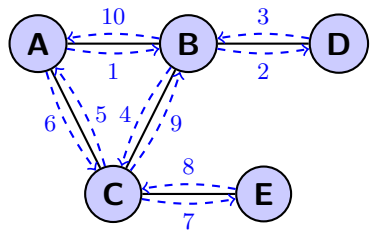

\section{Exploration by two agents}\label{sec: 2agents}

\subsection{Asynchronous exploration}
We now describe an exploration algorithm with two agents in the asynchronous model of Section \ref{sec:navigation-model}. The algorithm is the same for both agents which do not need to be distinguishable. It is presented in Table \ref{table:1} using the formalism of navigation tables introduced in Section \ref{sec:nav-table}. A brief intuition is as follows. The agents both perform independent versions of Trémaux algorithm, backtracking whenever they encounter previously explored vertices. The agents also use a new mark D (Done) to indicate that an edge has been traversed twice. When one of the agents returns to the origin and there is no adjacent unmarked passage, it uses the marker E left by the other agent to follow its trail. 
If one agent is only adjacent to passages marked D, it declares that the graph is explored. In this case, it will always be true that both agents are located on the same node. In practice, the second line of the navigation table, Table~\ref{table:1}, is never used in the asynchronous setting, but will reveal useful for the synchronous setting to account for the fact that the robots may decide to simultaneously traverse an unexplored edge in opposite directions.

\begin{table}[h!]
\centering
\begin{tabular}{||c c c c||} 
 \hline
    $p_u$ & $p_u'$ & $p_v$ & $p_v'$ \\ [0.5ex] 
 \hline\hline
 B & D & - & D \\ 
 $\emptyset$ & E & E & D \\
 $\emptyset$ & E & $v$ discovered & B \\
 $\emptyset$ & E & $v$ undiscovered & F \\
 F & D & - & D \\ 
 E & D & - & D \\ 
 \hline
\end{tabular}
\caption{Exploration with two agents: navigation table.}
\label{table:1}
\end{table}

\begin{theorem}\label{th:asynchronous}
In the asynchronous setting, the two agents using Table \ref{table:1} traverse all edges twice and then stop at the same node after exactly $2m$ rounds, where $m$ is the number of edges of the graph. 
\end{theorem}
We now turn to the analysis of the algorithm in the asynchronous setting. We make the following claims that lead to the proof of Theorem \ref{th:asynchronous}.

\begin{claim}\label{claim:number}
At the start of any round, the edges that have both passages unmarked have never been traversed, the edges which have both passages marked by \{B, E, F\} have been traversed once, the edges which have both passages marked D have been traversed twice. All of the edges in the graph fall in one of these cases.
\end{claim}
\begin{proof}
Initially, all passages are unmarked. When an unexplored edge is traversed for the first time, both of its endpoints are marked by one of B, E, F (recall that the second line of the table is not used in this section). Note that a marked passage never becomes unmarked in the course of the exploration. This proves that the edges that have never been traversed are exactly those having two unmarked passages. 

Also observe that whenever a passage marked B, E, F is traversed, both passages of the corresponding edge will be marked D. This proves the claim that the edges which have been traversed once are exactly those for which both passages are marked by B, E, F. 

Finally, observe that a passage marked D will never be used by the agent. This allows to conclude that the edges which have been traversed twice are exactly those with both endpoints marked D.
\end{proof}

\begin{claim}\label{claim:path}
    At the start of any round, the set of edges that have been traversed once form a path\footnote{that path may form a cycle if the two agents are located on the same node.} of disjoint edges between both agent's locations. 
\end{claim}
\begin{proof}
We provide a proof of this invariant by induction. 
The invariant holds initially, because all passages are unmarked and both agents are on the same node. We assume that the result holds at some round $t$, 
and we show that the result holds at the next round. We denote by $u_1\xleftrightarrow{e_1} \dots \xleftrightarrow{e_{\ell-1}} u_\ell$ 
the path of length $\ell \in \Nbb$ between both agents. We assume without loss of generality that the agent that is allowed to move at the present round is located at $u_1$.

If $u_1$ is adjacent to a passage marked B, then that passage must appear in the path between both agents. Thus moving the agent along that edge, and marking both of its endpoints by D preserves the invariant.
Else if $u_1$ is adjacent to an unmarked passage, then taking that passage and marking its endpoints by B, E, or F maintains the invariant. 
Else, $u_1$ is adjacent to a passage marked E which appears in the path between both agents. Thus moving the agent along that edge and marking both endpoints by D preserves the invariant.
\end{proof}

\begin{claim}\label{claim:finish}
    If an agent is adjacent only to passages marked D, both agents must be co-located and all edges have been visited twice.
\end{claim}
\begin{proof}
    Since the set of nodes that are marked B, E, F form a path between both agents, when one of the agents is adjacent only to passages marked D, it must be that the other is located on the same node and that all passages in the graphs have markers in \{$\emptyset$, D\}.

    We now assume by contradiction that one edge has never been traversed, i.e. that there is an unmarked passage. Without loss of generality, we consider a node that has been visited and that is adjacent to one such passage. This node cannot be the current position of the two agents. Consider the last time that an agent left this node in a move that was not a backtracking move. At that moment, the corresponding agent must have left the node through an unexplored edge, because those are always preferred to previously explored edges. Thus it left a marker E at that node, which forms a contradiction. 
\end{proof}
\begin{proof}[Proof of Theorem \ref{th:asynchronous}]
    By Claim \ref{claim:finish}, if the algorithm terminates, all edges have been visited at least twice and the agents are co-located. By Claim \ref{claim:number}, the algorithm must terminate in at most $2m$ steps, because each edge can be traversed at most twice. 
\end{proof}

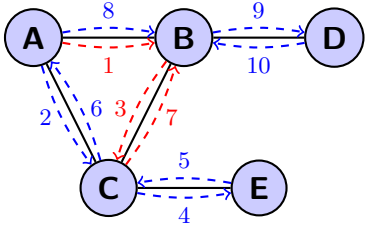
\begin{figure}
    \centering
    
\begin{tikzpicture}[
    thick,
    main node/.style={circle, fill=blue!20, draw, font=\sffamily\large\bfseries}
  ]

  \node[main node] (A) at (0,0) {A};
  \node[main node] (B) at (2,0) {B};
  \node[main node] (C) at (1,-2) {C};
  \node[main node] (D) at (4,0) {D};
  \node[main node] (E) at (3,-2) {E};

  \draw (A) -- (B) -- (D);
  \draw (A) -- (C) -- (B);
  \draw (C) -- (E);

  \draw[red, dashed, ->] (A) edge[bend right=10] node[midway, below, font=\small] {1} (B);
  \draw[blue, dashed, ->] (A) edge[bend right=10] node[midway, left,font=\small] {2} (C);
  \draw[red, dashed, ->] (B) edge[bend right=10] node[midway, left,font=\small] {3} (C); 
  \draw[blue, dashed, ->] (C) edge[bend right=10] node[midway, below,font=\small] {4} (E);
  \draw[blue, dashed, ->] (E) edge[bend right=10] node[midway, above,font=\small] {5} (C);
  \draw[blue, dashed, ->] (C) edge[bend right=10] node[midway, right,font=\small] {6} (A);
  \draw[red, dashed, ->] (C) edge[bend right=10] node[midway, right,font=\small] {7} (B);
  \draw[blue, dashed, ->] (A) edge[bend left=10] node[midway, above,font=\small] {8} (B);
  \draw[blue, dashed, ->] (B) edge[bend left=10] node[midway, above,font=\small] {9} (D);
  \draw[blue, dashed, ->] (D) edge[bend left=10] node[midway, below,font=\small] {10} (B);
\end{tikzpicture}

\caption{Ariadne (blue) and Theseus (red), exploring a graph with $m=5$ (asynchronous setting).}
    \label{fig:AandT}
\end{figure}

\subsection{Synchronous exploration}\label{sec:synchronous}
We now consider the synchronous setting, in which both agents move simultaneously at all rounds. We shall not change the algorithm based on Table \ref{table:1}, except for when both agents are co-located. 
In this case, we will simply assume that both steps \eqref{step1} occur sequentially, instead of simultaneously, in order to avoid the situation in which both agents would always choose the same ports and would thus never split. This is easily enforced by the assumption that the agents may coordinate when they are located at the same node. Such assumption is always granted in collective exploration \cite{fraigniaud2006collective,brass2011multirobot}. 

\begin{theorem}\label{th:explo-synchronous}
    In the synchronous setting, the two agents using Table \ref{table:1} traverse all edges of the graph and meet at some node in exactly $m$ time-steps, where $m$ is the number of edges. 
\end{theorem}

\begin{proof}
    We show that Claim \ref{claim:number} and Claim \ref{claim:path} also hold in the synchronous setting. 

\begin{proof}[Proof of Claim \ref{claim:number} in the synchronous setting] The proof is the same as in the sequential case, except for the particular situation where both agents move simultaneously along the same edge in opposite directions. Note that this situation only occurs if both passages of some edge $u\xleftrightarrow{e}v$ are unmarked. In this case, the edge is traversed twice in a single round. Fortunately, the second line of Table \ref{table:1} allows to catch this situation and both passages of the edge are marked D at the end of this time-step. \end{proof}

\begin{proof}[Proof of Claim \ref{claim:path} in the synchronous setting]
The proof works again by induction and the arguments above still hold except for the case of agents moving simultaneously along the same edge in opposite directions. Observe that such situation may be seen as equivalent to two consecutive moves of a single agent in the asynchronous model. In the first move, the moving agent meets the other agent and mark the passage $u\xrightarrow{e} v$ by $E\xrightarrow{e} B$. In the subsequent move, that agent backtracks on the same edge which gets marked $D\xleftrightarrow{e} D$. Note that the final configuration is the same as for the synchronous setting in which both agents move along $e$ simultaneously in opposite directions, except that the position of the agents in the final configuration is inverted, which is without consequence for the rest of the execution. The invariant is thus also preserved in this situation. 
\end{proof}
Finally, the proof of Claim \ref{claim:finish} is unchanged.
\end{proof}

\subparagraph{Competitive ratio}\label{sec:implicationcr}
We now briefly discuss how the result of Theorem~\ref{th:explo-synchronous} improves the competitive ratio of collective graph exploration. In collective graph exploration, it is clear that $\texttt{OPT}(G,k) \geq m/k$ where $\texttt{OPT}(G,k)$, denotes the minimum number of rounds required by the team to traverse all edges of $G$, if the team has full knowledge of $G$. Our algorithm, which satisfies $\texttt{ALG}(G,k)\leq m$, for any $k\geq 2$, thus improves the competitive guarantee of the single depth first-search (requiring $2m$ time-steps), from $2k$ to $k$, partially answering the question in the conclusion of \cite{brass2014improved}.  

In a more restrictive formulation of collective exploration, the agents are required to return to the origin after all edges have been traversed. In that case, it is generally assumed that the agent have unbounded memory and computation \cite{disser2017general,brass2014improved}. After $m$ synchronous time-steps, the agents meet at some node of the graph and have enough information to compute a shortest path leading them back to the origin.  In this case our guarantee becomes $\texttt{ALG}(G,k) \leq m+D$, where $D$ is the diameter of the graph. For this variant of the problem, we also have $\texttt{OPT}(G,k) \geq m/k$ and $\texttt{OPT}(G,k) \geq 2D$. Thus, $\texttt{OPT}(G,k) \geq \max\{m/k,2D\}\geq \frac{k}{k+1/2}m/k + \frac{1/2}{k+1/2}2D \geq \frac{1}{k+1/2}(m+D)= \frac{1}{k+1/2} \texttt{ALG}(G,k)$. In this formulation of the problem, we improve the competitive ratio of the problem from $2k$ to $k+1/2$, for any $k\geq 2$.

\section{Rendezvous of two agents}\label{sec:rdv}
\subsection{Asynchronous rendezvous}
In this section, we present a rendezvous algorithm for the asynchronous model. The algorithm differs from most of the literature because we allow the two agents to use different algorithms. 

\subparagraph{Algorithm.} The distinct agents are called Ariadne and Theseus, and are defined by their navigation tables, Table~\ref{table:rdv-Ariadne} and Table \ref{table:rdv-theseus}. We assume that they can sense when they are located on a same node, and that they stop if that is the case. A brief intuition on the algorithm is as follows. Both agents run a depth-first search until they land on a node previously discovered by the other agent. Then, they follow the trail left by the other agent. To avoid both agents from cycling behind each other indefinitely, one of the agent (Ariadne) will retrace her path when she realises that the other (Theseus) is following her trail. We state the main result. 
\begin{proposition}
    Ariadne and Theseus, defined respectively by Table \ref{table:rdv-Ariadne} and Table \ref{table:rdv-theseus} meet in any unknown graph $G$ in at most $3m$ steps, where $m$ is the number of edges in $G$. 
\end{proposition}

\begin{table}[h!]
\centering
\begin{tabular}{||c c c c||} 
 \hline
    $p_u$ & $p_u'$ & $p_v$ & $p_v'$ \\ [0.5ex] 
 \hline\hline
 
 B\textsubscript{T} & D & - & D \\
 E\textsubscript{A} & E\textsubscript{AT} & - & F\textsubscript{AT} \\
 E\textsubscript{AT} & E\textsubscript{ATT} & - & F\textsubscript{ATT} \\ 
 $\emptyset$ & E\textsubscript{T} & $v$ discovered by T & B\textsubscript{T} \\
 $\emptyset$ & E\textsubscript{T} & $v$ undiscovered by T & F\textsubscript{T} \\
 F\textsubscript{T} & D & - & D \\ 
 \hline
\end{tabular}
\caption{Theseus.}
\label{table:rdv-theseus}
\end{table}

\begin{table}[h!]
\centering
\begin{tabular}{||c c c c||} 
 \hline
    $p_u$ & $p_u'$ & $p_v$ & $p_v'$ \\ [0.5ex] 
 \hline\hline
 
 B\textsubscript{A} & D & - & D \\
 E\textsubscript{T} & E\textsubscript{AT} & - & F\textsubscript{AT} \\
 F\textsubscript{AT} & D' & - & D' \\
 $\emptyset$ & E\textsubscript{A} & $v$ discovered by A & B\textsubscript{A} \\
 $\emptyset$ & E\textsubscript{A} & $v$ undiscovered by A & F\textsubscript{A} \\
 F\textsubscript{A} & D & - & D \\ 
 \hline
\end{tabular}
\caption{Ariadne.}
\label{table:rdv-Ariadne}
\end{table}

\subparagraph{Termination.} The algorithm terminates if the two agents are on the same node, or alternatively if one agent does not have a possible move. We start with the following claim, which can be directly verified from the navigation tables. 
\begin{claim}
The edges with an endpoint marked by 
\{B\textsubscript{A}, 
E\textsubscript{A}, 
F\textsubscript{A}, 
B\textsubscript{T}, E\textsubscript{T}, F\textsubscript{T}\}
have been traversed once. 
The edges with an endpoint marked by
\{E\textsubscript{AT}, F\textsubscript{AT}, 
D\} have been traversed twice. 
The edges with an endpoint marked by
\{E\textsubscript{ATT}, F\textsubscript{ATT}, D'\} have been traversed three times.
\end{claim} 
It is clear from the preceding claim that no edge gets traversed more than $3$ times. Thus the algorithm terminates in at most $3m$ steps. 

\subparagraph{Correctness.} We now show that algorithm is correct, i.e. that if the moving agent does not have an adjacent port matching the first column of his navigation table, then it is located on the same node as the other agent. 

While Ariadne and Theseus do not walk on a node discovered by the other, they both perform a Trémaux depth-first search, using letter D to indicate that some edge has been used twice and will be discarded. 
Recall from the the analysis of the previous section that in this phase, the edges with a tail E\textsubscript{A} form a directed path from the origin of Ariadne leading to the current position of Ariadne and the edges with a tail E\textsubscript{T} form a directed path from the origin of Theseus to the current position of Theseus. 

The first round when Theseus is located on a node initially discovered by Ariadne, this node must be adjacent to a port marked E\textsubscript{A} because its exploration was not finished by Ariadne before the move of Theseus, and she has not been back since (otherwise the rendezvous would be complete). From that moment, the behaviour of Theseus is that it will follow the markers E\textsubscript{A} or E\textsubscript{AT} which lead to Ariadne. 

The first round when Ariadne is located on a node initially discovered by Theseus, this node must be adjacent to a port marked E\textsubscript{T}. Ariadne will use this marker to follow the trail of Theseus. When this leads her to a node that she had initially discovered, the marker E\textsubscript{T} is replaced by marker E\textsubscript{AT} and she understands that Theseus is also following her. She thus retrace her steps using markers F\textsubscript{AT} and F\textsubscript{A} to get to Theseus.

We now formalize the arguments above by the following claims.

\begin{claim}
At all rounds, the edges with tail in \{E\textsubscript{A}, E\textsubscript{AT}, E\textsubscript{ATT}\} form a directed path from the origin of Ariadne to the position of Ariadne.
\end{claim}
\begin{proof}
    The statement is verified by induction on the moves of Ariadne similarly to the proof of Claim \ref{claim:path}. Observe that Ariadne always leaves a marker E\textsubscript{A} or E\textsubscript{AT} behind her, except when she backtracks and closes some edge. Also, Theseus can only convert makers E\textsubscript{A} and E\textsubscript{AT} in markers E\textsubscript{AT} and E\textsubscript{ATT}. This suffices to prove the claim.
\end{proof}

\begin{claim}\label{claim:pathTtoA} If Theseus has found a node earlier discovered by Ariadne,  the edges with tail in \{E\textsubscript{A}, E\textsubscript{AT}\} lead from Theseus to Ariadne.
\end{claim}
\begin{proof}
We consider the first round when Theseus lands on some node $v$ initially discovered by Ariadne. 
Notice that the exploration of this node was not finished by Ariadne at this round, thus it must be adjacent to some marker E\textsubscript{A}. Thus $v$ belongs to the directed path of edges with tail in \{E\textsubscript{A}, E\textsubscript{AT}, E\textsubscript{ATT}\} which goes from the origin to Ariadne. Since Theseus has not yet traversed an edge traversed earlier by Ariadne, there can be no markers E\textsubscript{ATT} in the graph. At this instant, it is therefore the case that edges with tail in  \{E\textsubscript{A}, E\textsubscript{AT}\} lead from Theseus to Ariadne, and it is clear that this property is preserved by all subsequent moves of Theseus (which will be through edges marked E\textsubscript{A} or E\textsubscript{AT}).
\end{proof}

\begin{claim}\label{claim:pathAtoT}
If Theseus has not found a node earlier discovered by Ariadne, but Ariadne has found a node earlier discovered by Theseus, the edges with tail in E\textsubscript{T} lead from Ariadne to Theseus.
\end{claim}
\begin{proof}
While both agents have disjoint itineraries, the set of edges with tails E\textsubscript{T} form a directed path from the origin of Theseus to the position of Theseus. The first round when Ariadne is located on a node discovered earlier by Theseus, the edges with tail marked E\textsubscript{T} thus lead from Ariadne to Theseus. This invariant is preserved for any move of Theseus or Ariadne, for as long as Theseus does not find a node discovered by Ariadne. 
\end{proof}

\begin{claim}
At least one of Theseus and Ariadne will find a node earlier discovered by the other.
\end{claim}
\begin{proof}
While they do not find a node discovered by the other, it is clear from the tables that both Ariadne and Theseus each run an instance of depth-first search. Since the graph is connected, and the adversary must move one of the agents at each round, the algorithm cannot stop before one of the agent has found a node discovered by the other. 
\end{proof}

\begin{claim}\label{claim:finalrdv}
If one of Ariadne or Theseus does not have a move prescribed by their navigation table, they are located on the same node.
\end{claim}
\begin{proof} 
We assume that one of Ariadne or Theseus does not have a move prescribed by their navigation table. By previous claim, it must be the case that either (1) Theseus has found a node earlier discovered by Ariadne or (2) Theseus has never found a node earlier discovered by Ariadne but Ariadne has found a node earlier discovered by Theseus. 

(1) In the first case, if the agents are not co-located by Claim \ref{claim:pathTtoA} there are markers in \{E\textsubscript{A}, E\textsubscript{AT}\} adjacent to Theseus, and markers in \{F\textsubscript{A}, F\textsubscript{AT}, B\textsubscript{A}\} adjacent to Ariane, thus both agents have a possible move in their navigation table.

(2) In the second case, by Claim \ref{claim:pathAtoT} if the agents are not co-located, then there is a marker E\textsubscript{T} adjacent to Ariadne and markers in \{F\textsubscript{T},  B\textsubscript{T}\} adjacent to Theseus. Thus, both agents have a possible move in their navigation table. 

The agents are thus on the same node when the algorithm stops. 
\end{proof}

Claim \ref{claim:finalrdv} ends the formal proof of correctness, by showing that both agents have a prescribed move in their navigation table, at least until the rendezvous.

\begin{remark} We notice that the total number of moves before rendezvous can be reduced from $3m$ to $2m+n-1$. It suffices to observe that the set of edges which may be traversed three times are initially marked E\textsubscript{T}$\leftrightarrow$F\textsubscript{T} and that there can be at most $n-1$ such edges. 
\end{remark}
\subsection{Synchronous rendezvous}
We now translate the above algorithm to the synchronous setting. We shall not change the navigation tables, but for simplicity we assume that rendezvous is achieved if Theseus and Ariadne travel through the same edge in opposite directions at the same synchronous round. This assumption is easily relaxed in Remark \ref{rem-rdv}.

\begin{theorem}\label{th:rdv-synchronous}
In the synchronous setting, Ariadne and Theseus achieve rendezvous in at most $\ceil{3m/2}$ time-steps, where $m$ is the number of edges of the graph.
\end{theorem}

\begin{proof}
It suffices to observe that if Theseus and Ariadne have not achieved synchronous rendezvous in $t$ time-steps, the state of the markers in the graph and the position of the agents is the same as that obtained by a run of the asynchronous setting for $2t$ moves. Indeed, assume that the property is true at time step $t$. The agents then synchronously choose an adjacent port. Either the agents will achieve rendezvous at that round, or it is possible to asynchronously first move one of the agent to its destination, which is not the position of the other agent, and only then to move the other agent at its destination. This proves the property at time $t+1$ and finishes the proof of Theorem \ref{th:rdv-synchronous}.
\end{proof}

\begin{remark}\label{rem-rdv} The assumption that rendezvous is achieved when the agents meet inside an edge is easy to relax. Indeed, it suffices to note that if both agents (using Table \ref{table:rdv-Ariadne} and \ref{table:rdv-theseus}) traverse one edge in opposite directions at the same time-step the markers that they observe at the other endpoint of the edge will be inconsistent with the marker that they observe at the initial endpoint of the edge (because the edge has been used one extra-time in the meanwhile). Therefore, both agent realize that they have just traversed the same edge in both directions. One of the agents (say, Theseus) can stop, while the other (say, Ariadne) backtracks. Rendezvous in this sense just requires one extra time-step.
\end{remark}


\section{Generalizations}\label{sec:discussion}
In this section, we show that our algorithms immediately adapt to natural extensions of the mobile computing model defined in Section \ref{sec:navigation-model}, and providing matching lower-bounds in this setting.
 
 \subsection{Weighted graphs}\label{sec:weighted}
We observe that the problem of collective exploration or rendezvous can be extended to weighted graphs. For any edge $e\in E$, we denote by $w_e$ the weight (or length) of that edge, which corresponds to the cost paid by an agent which traverses $e$. We also denote by $L = \sum_{e\in E}w_e$ the sum of all total edge lengths. The goal of collective graph exploration (resp. rendezvous) then becomes to traverse all edges of the graph (resp. to meet somewhere in the graph) while paying a limited total cost. Since the above proofs actually bound the maximum number times each edge is traversed by $2$ in the case of exploration (resp. by $3$ in the case of rendezvous), we immediately have the following results. 

\begin{proposition}\label{prop:explo-weighted}
    The two agents using Table \ref{table:1} perform asynchronous collective exploration in weighted graph, while paying a total cost of at most $2L$.
\end{proposition}

\begin{proposition}\label{prop:rdv-weighted}
    The two agents using Table \ref{table:rdv-Ariadne} and Table \ref{table:rdv-theseus} perform asynchronous rendezvous in an unknown weighted graph, while paying a total cost of at most $3L$.
\end{proposition}

\subsection{Continuous moves}\label{sec:continuous}
We now support the claim of the introduction that our model of mobile computing captures the situation where the agent move continuously and the the adversary controls their speeds at all times. 
For this, it is useful and natural to assume that the agents can communicate if they meet inside a given edge -- because one agent could be blocked indefinitely inside an edge. We now make the following claim.

\begin{claim}\label{claim:continous}For any algorithm based on a navigation table, while the robots do not meet inside an edge, the setting where the agents move continuously (and the adversary controls their speed) is equivalent to a run of the setting of Section \ref{sec:navigation-model} where the agents make discrete moves (and the adversary chooses which robot moves).     
\end{claim} 

\begin{proof} 
Without loss of generality, assume that all edges are split into two sub-edges, with a node in the middle. We will consider the exploration and rendezvous problem in this new graph. Observe that our guarantees are preserved, because the sum of all edges lengths $L$ is invariant by this transformation (see Section \ref{sec:weighted} for the definition of $L$). We now let the adversary control the (continuous) speeds of the agents at all times until they meet inside an edge of the original graph at some time $t$. We say that an event is \textit{triggered} at every instant when an agent leaves a node of the original graph (i.e. not a midway node), or attains a node in the original graph. 
 We denote by
$t^{(A)}_1,t^{(A)}_2, \dots, t^{(A)}_{m_A}$ all the instants (before $t$) at which Ariadne triggers an event. And we define $t^{(T)}_1,  t^{(T)}_2, \dots, t^{(T)}_{m_T}$ similarly for Theseus. We also define $r_i\in \{A,T\}$ for $i\in \{1,\dots, m_A + m_T\}$ to be the index of the agent that triggers the $i$-th event. We then observe (by an immediate induction in $m_A$ and $m_T$) that the sequence of nodes attained by the agents in the continuous model correspond exactly to the sequence of nodes that would have been attained by the agents if the adversary had attributed the following sequence of discrete moves: $r_1,r_2,\dots, r_{m_A+m_T}$.
\end{proof}
Note that in this model of continuous moves, the \textit{cost} of exploration (resp. rendezvous) becomes equal to the total amount of energy that was afforded to the agents by the adversary before termination, where a unit of energy can be used by an agent to move by a unit distance. We then have the following proposition.  

\begin{proposition}
The guarantees of Proposition \ref{prop:explo-weighted} and Proposition \ref{prop:rdv-weighted} generalize to the setting where the agents move continously (and the adversary controls their speeds).
\end{proposition}

\begin{proof}
    By the definition of the problem of rendezvous, the task is completed when the agents meet inside an edge. Therefore, by Claim \ref{claim:continous}, our guarantees of the model with discrete moves transfer to the model with continuous moves. The argument is similar for exploration since the algorithm based on Table \ref{table:1} must have completed exploration whenever the two agents meet inside an edge. 
\end{proof}

\subsection{Lower bounds}
We now show that our guarantees are met with matching lower-bounds in the setting of Section \ref{sec:continuous}. The fact that the graphs are weighted and that energy/movement is attributed by an adversary is not crucial to the demonstration, but it greatly simplifies the proof (especially for rendezvous). In the following statement, $L$ denotes the sum of all edge lengths in a weighted graph. 

\begin{proposition} Any exploration algorithm with two agents of cost bounded by $\alpha L$ in weighted graphs satisfies $\alpha \geq 2$. Any deterministic asynchronous rendezvous algorithm with two agents cost bounded by $\alpha L$ in weighted graphs satisfies $\alpha \geq 5/2$ and $\alpha \geq 3$ if the agents have finite memory.
\end{proposition}
\begin{proof}
    The exploration lower-bound is simple. Consider a line of length $L$. If both agents receive the same speed starting from one end of the line, they (collectively) spend an energy of $2L$ to traverse it. 

    For rendezvous, we consider two classes of graphs illustrated in Figure \ref{fig:cycle-broken}: the cycle graphs and the `broken cycle' graphs. In the figure, the starting points of Ariadne and Theseus are denoted by A and T. Importantly, the lengths of all edges are arbitrary and are initially unknown to the agents. We consider an adversary which starts by letting both agents move along one edge. The agents have not better choice than to move along an arbitrary edge (not moving is a costly option, because the adversary could provide energy only to the immobile agent indefinitely). Since the algorithm is deterministic, we can assume that the two agents are in the circle and have moved along opposite edges. At this stage, the agents do not know whether the graph is a cycle or a broken cycle, they merely know the length of the edge that they have traversed, and that they attained the initial position of the other agent. The adversary then provides more energy to both agent. The agents have no better choice than to choose the edge which they have not yet traversed: indeed, if they are in the broken cycle, this is the only way to get to the other agent without risking to traverse one arbitrary edge three times, which would immediately lead to $\alpha \geq 3$. At this point, both agents are back on their initial position and realize that the graph is a cycle. If the agents are able to compare edge lengths (because they have enough memory to store large numbers), the agents can choose to meet inside the shortest edge. In this case, the worst situation is when both edges have the same length $\ell$, leading to a total cost of $5\ell = 5L/2$. If the agents do not have enough memory to compare edge lengths, then they must pick one arbitrary edge on which to meet, using their identifier to break the tie. The adversary will make sure that this edge is the longest (by far), and since it gets traversed three times we have $\alpha \geq 3$.
\end{proof}
\begin{remark}We note that the lower-bound of $5/2 L$ is matched by a variant of our rendezvous algorithm (which uses unbounded memory). Denoting by $u$ the first node that was discovered by Ariadne on the trail of Theseus, and by $v$ the first node that was discovered by Theseus on the trail of Ariadne, we observe that the trails of Theseus and Ariadne form a cycle containing $u$ and $v$. Once an agent has identified both $u$ and $v$ it could decide to traverse the shortest of both paths between $u$ and $v$ (andpossibly continuing its way on the cycle until it meets the other agent).    
\end{remark}

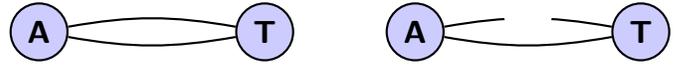
\begin{figure}
    \centering
\begin{tikzpicture}[
    thick,
    main node/.style={circle, fill=blue!20, draw, font=\sffamily\large\bfseries}
  ]

  \node[main node] (A) at (0,0) {A};
  \node[main node] (T) at (3,0) {T};
    
    \draw (A) edge[bend right=10] node[midway, below, font=\small] {} (T);
  \draw (A) edge[bend left=10] node[midway, above,font=\small] {} (T);

  \node[main node] (Ap) at (5,0) {A};
  \node[main node] (Tp) at (8,0) {T};
    
    \draw (Ap) edge[bend right=10] node[midway, below, font=\small] {} (Tp);
  \draw (Ap) edge[bend left=10] node[midway, above,font=\small] {} (Tp);

    \node[circle,fill=white!20] (W) at (6.35,0.2) {};
    \node[circle,fill=white!20] (W) at (6.65,0.2) {};
\end{tikzpicture}
\caption{The cycle and the broken cycle, edge lengths are adversarial and unknown.}
    \label{fig:cycle-broken}
\end{figure}

\subsection*{Conclusion}
In this paper, we studied two problems of mobile computing, exploration and rendezvous, for two agents in an unknown graphs. For both problems, we provide algorithms that improve over the naive strategies that are based on depth-first search. Our guarantees hold for a general model of asynchrony, and generalize to weighted graphs for which they have matching lower-bounds.

\subsection*{Acknowledgements}
The author thanks Andrzej Pelc for his valuable feedback and suggestions, as well as the Argo team at Inria.

\bibliography{sample}

\end{document}